\documentclass[a4paper,UKenglish,cleveref, autoref, ]{lipics-v2019}

\nolinenumbers

\title{Generalized Assignment via Submodular Optimization with Reserved Capacity} 

\titlerunning{Generalized Assignment via Submodular Optimization}

\author{Ariel Kulik}{Computer Science Department, Technion, Haifa, Israel}{{kulik@cs.technion.ac.il}}{}{}

\author{Kanthi Sarpatwar}{IBM Research, Yorktown Heights, NY, USA}{{sarpatwa@us.ibm.com}}{}{}

\author{Baruch Schieber}{Department of Computer Science, New Jersey Institute of Technology, Newark, NJ, USA}{{sbar@njit.edu}}{}{}

\author{Hadas Shachnai}{Computer Science Department, Technion, Haifa, Israel}{hadas@cs.technion.ac.il}{}{}

\authorrunning{A.\,Kulik, K.\,Sarpatwar, B.\,Schieber and H.\,Shachnai}

\Copyright{Ariel Kulik, Kanthi Sarpatwar, Baruch Schieber and Hadas Shachnai}



\category{}

\relatedversion{}

\supplement{}

\acknowledgements{H. Shachnai's work was conducted during a visit to DIMACS
	partially supported by
	the National Science Foundation under grant number CCF-1445755.}

\usepackage{amsmath}
\usepackage{amsfonts}
\usepackage{amsthm}

\usepackage{algorithm}
\usepackage{algorithmicx}
\usepackage{algpseudocode}

\newcommand{\negA}{\vspace{-0.05in}}

\newcommand{\mysubsection}[1]{\negA\subsection{#1}\negA}

\usepackage{amssymb}

\DeclareMathOperator*{\argmax}{arg\,max}
\DeclareMathOperator*{\argmin}{arg\,min}

\newcommand{\LP}{\mbox{LP}}

\newcommand{\set}[1]{\{#1\}}
\newcommand{\eps}{\varepsilon}
\newcommand{\remove}[1]{}

\newcommand{\OPT}{\mbox{OPT}}

\newcommand{\cG}{\cal G}
\newcommand{\ccG}{\Omega}
\newcommand{\f}{\phi}
\newcommand{\ft}{\eta}
\newcommand{\g}{\psi}
\newcommand{\ca}{\cal A}
\newcommand{\ggap}{\sf Group GAP}
\newcommand{\gap}{\sf GAP}
\newcommand{\Reals}{\rm I\!R}
\newcommand{\Int}{\mathbb{Z}}
\newcommand{\hsi}{{\hat s}_i}
\newenvironment{dl_proof}[1]{\vspace{\baselineskip}\noindent{\bf Proof of Lemma #1:}}{
	\hspace*{\fill}{\qed}}

\begin{document}

\maketitle

 \begin{abstract}
	
We study a variant of the {\em generalized assignment problem} ({\gap})
with group constraints.
An instance of {\ggap} is a set $I$ of items, partitioned into $L$ groups,
and a set of $m$ uniform (unit-sized) bins.
Each item $i \in I$ has a size $s_i >0$, and a profit
$p_{i,j} \geq 0$ if packed in bin $j$.
A group of items is {\em satisfied} if all of its items are packed.
The goal is to find a feasible packing of a subset of the items in the bins
such that the total profit from satisfied groups is maximized.
We point to central applications of {\ggap} in Video-on-Demand services, 
mobile Device-to-Device network caching and base station cooperation in
5G networks.

Our main result is a $\frac{1}{6}$-approximation algorithm for {\ggap} instances
where the total size of each group is at most $\frac{m}{2}$.
At the heart of our algorithm lies 
an interesting derivation of a submodular
function from the classic LP formulation of
{\sf GAP}, which facilitates the construction of
a high profit solution utilizing
at most half the total bin capacity, while the other half is 
{\em reserved} for later use.
In particular, we give an algorithm for submodular maximization
subject to a knapsack constraint,
which finds a solution of profit at least $\frac{1}{3}$ of the optimum, using
at most half the knapsack capacity, under mild restrictions on element sizes.
Our novel approach of submodular optimization 
subject to a knapsack {\em with reserved capacity} constraint
may find applications in solving other group assignment problems.

\end{abstract} 
\clearpage
\pagenumbering{arabic}
\section{Introduction}
\label{sec:intro}

With the rapid adoption of cloud computing, wireless networks, and other modern platforms, resource
allocation problems of various flavors have regained importance. One
classic
example is the {\em generalized assignment problem} ({\gap}).
We are given a set of $n$ items and $m$ bins, $[m]=\{1, 2, \ldots , m\}$. Each item
 $i\in [n]$ has a size $s_{i,j} > 0$ and a profit $p_{i,j} \geq 0$ when packed into bin $j \in [m]$.
 The goal is to feasibly pack in the bins a subset of the items
of maximum total profit.
{\sf GAP} has been widely studied, with applications ranging from
grouping and loading for manufacturing systems to  land use optimization in
regional planning
(see, e.g., \cite{B76,CH99}). In discrete optimization, {\sf GAP}
has received considerable attention
also as  a special
case of the {\em separable assignment} problem  and
 {\em submodular maximization}~(see, e.g., \cite{V08,FGMS11,BTW15,BF16}).
We consider a 
variant of {\gap} with group constraints.
An instance of {\ggap} consists of
 a set $I=\{1,2, \ldots, n\}$ of $n$ items and $m$ uniform (unit-sized) bins $M=\{1, \ldots , m \}$. Each item $i \in I$ has a size $ s_i >0 $ and a profit
$p_{i,j} \geq 0$ when assigned to bin $1 \leq j \leq m$.
The items in $I$ are partitioned into $L \geq 1$ groups,  ${\cG} = \{ G_1, \ldots, G_L \}$.
Given an assignment of items to bins, we say that a group is {\em satisfied} if all of
its items are assigned. 
The goal is to find a feasible assignment of a subset of the items to bins
such that the total profit from satisfied groups is maximized.
Formally, a feasible assignment is a tuple $(U_1, \ldots , U_m)$, such that $U_j \cap U_k = \emptyset$ for all $1 \leq j < k \leq m$,
$U_j \subseteq I$ and $\sum_{i \in U_j} s_i \leq 1$, for all $1 \leq j \leq m$. Let $I(U) = \cup_j U_j$. Then $G_\ell \in {\cG}$ is satisfied if $G_\ell \subseteq I(U)$.
Let ${\cG}_s = \{ G_{\ell_1}, \ldots G_{\ell_t} \}$ be the set of satisfied groups and $I({\cG}_s) = \cup_{G_\ell \in {\cG}_s } G_\ell$.
Then we seek an assignment $(U_1, \ldots , U_m)$ for which $\sum_{j=1}^m \sum_{i \in U_j \cap I({\cG}_s )} p_{i,j}$ is maximized. 

The following scenario suggests a natural application for
{\ggap}. Consider a Video-on-Demand (VoD) service where each  video is given as a collection of segments.
The system has a set of $m$ servers of uniform capacity
distributed over multiple locations. 
To obtain revenue from a video the system must store 
all of its segments (possibly  on different  servers). The revenue from a specific
video also depends on the servers which store  the 
segments. This is  due to the content delivery costs 
resulting from the distance between
the servers  and the predicted location 
of the video audience. The objective of the 
VoD service provider is to select a subset of segments
and an allocation of these segments to servers so as to maximize the total revenue.
In the Appendix we describe central applications of {\ggap} in
mobile Device-to-Device network caching and in base station cooperation in 5G networks.

\mysubsection{Prior Work}
\label{sec:prior_work}

We note that a {\ggap} instance in which each group consists of a {\em single} 
item yields an instance of classic {\gap} where each item takes a single size
across the bins, and all the bins have identical capacities.
{\sf GAP} is known to be APX-hard already in this case, even if there  are only two possible item sizes, and
each item can take one of two possible profits~\cite{ck00}.
Thus, most of the previous research focused on obtaining
efficient approximate solutions.\footnote{Given an algorithm ${\ca}$, let ${\ca}(I),OPT(I)$ denote the
	profit of the solution output by ${\ca}$ and by an optimal solution for a problem instance $I$, respectively.
	For $\rho \in (0,1]$, we say that
	${\ca}$ is a $\rho$-approximation algorithm if, for any instance $I$, $\frac{{\ca}(I)}{OPT(I)} \geq \rho$.}
Fleischer et al.~\cite{FGMS11} obtained a $(1 - e^{-1})$-approximation
for {\sf GAP}, as a special case of the {\em separable assignment problem}.
Feige and Vondr{\'a}k~\cite{FV06} obtained the current best known ratio of $1 - e^{-1}+\eps$, for some absolute constant $\eps >0$.

Chen and Zhang~\cite{CZ16} studied the problem of {\em group packing of items into multiple knapsacks} ({\sf GMKP}), a
 special case of {\ggap} where
the profit of each item is the same across the bins.
Let {\sf GMKP}$(\delta)$ be 
the restriction of {\sf GMPK} to instances in which the total size of items in each group is at most $\delta m$ (that is, a factor $\delta$ of the total capacity of all bins).  
For $\delta >\frac{2}{3}$, the paper~\cite{CZ16} rules out the existence of a constant factor approximation for {\sf GMKP}$(\delta)$, unless {\sf P}= {\sf NP}. 
For $\frac{1}{3}<\delta \leq \frac{2}{3}$, the authors show that there is no 
$\left(\frac{1}{2} +\eps\right)$-approximation for  {\sf GMKP}$(\delta)$, unless {\sf P}= {\sf NP}, and derive a nearly matching $\left(\frac{1}{2}- \eps\right)$-approximation, for any $\eps >0$. 
The paper presents also approximation algorithms and hardness results for other special cases of {\sf GMPK}.

There has been earlier work also on variants of {\ggap} 
with the added constraint that
in any feasible assignment there is at most one item from $G_\ell$ in bin $j$, for any
$1 \leq \ell \leq L$, $1 \leq j \leq m$. Adany et al.~\cite{AF+16} considered this problem, 
called {\em all-or-nothing} {\sf GAP} {\sf (AGAP}). They presented
a $(\frac{1}{19} - \eps)$-approximation algorithm for general instances, and a 
$(\frac{1}{3} -\eps)$-approximation for the special case where the profit of an item is identical across the bins, called the {\em group packing} ({\sf GP}) problem.
Sarpatwar et al.~\cite{SSS18} consider a more general setting for {\sf AGAP}, where each group of items is associated with a time window in which it can be packed.
The paper shows that this variant of the problem, called $\chi$-{\sf AGAP}, admits an $\Omega(1)$-approximation, assuming the time windows are large enough relative
to group sizes. Specifically, for a group $G_\ell$ having a time window of $m$ slots (= $m$ bins), it is assumed that $s(G_\ell) \leq \frac{m}{20}$.

\subsubsection{Submodular Maximization}

Given a finite set $\Omega$, a function $f:2^{\Omega} \rightarrow \Reals$ 
is {\em submodular} if for every $S,T \subseteq \Omega$ we have 
\[f(S) + f(T) \geq f(S \cup T) + f(S\cap T).\]

An equivalent definition of submodularity refers to its diminishing returns:
$$f(S \cup \{ u \}) - f(S) \leq f(T \cup \{ u \} )- f(T),$$ for any $T \subseteq S \subseteq \Omega$, and $u \in \Omega \setminus S.$
A set function $f$ is {\em monotone} if for every $S\subseteq T \subseteq \Omega$
it holds that $f(S) \leq f(T)$. 
Submodular functions arise naturally in a wide variety of optimization problems, ranging from
coverage problems and graph cut problems to {welfare problems}
(see~\cite{BF17} for a survey on submodular functions). 
Submodular optimization under various constraints has been widely studied in the past four decades (see, e.g.,~\cite{S04, CCPV11, FNS11} and~\cite{BF17} and the references therein).

The problem of maximizing a monotone submodular function subject to a knapsack constraint is defined as follows. We are given an oracle to a  monotone, non-negative submodular function 
$f:2^{\Omega} \rightarrow \mathbb{R}$.  Each element $i\in \Omega$  is associated with 
a size $s_i \geq 0$. We are also given a capacity $B > 0$. The objective is to find 
a subset $S\subseteq \Omega$ such that $\sum_{i \in S} s_i \leq B$ and 
$f(S)$ is maximized. The best known result is a $(1 - e^{-1})$-approximation algorithm due to Sviridenko~\cite{S04}. The ratio of $(1-e^{-1})$ cannot be improved even
when $f$ is a coverage function and element sizes are
uniform, unless {\sf P}= {\sf NP}~\cite{Fe98}. A matching lower bound of $(1-e^{-1})$ is known also
for the oracle model with no complexity assumption~\cite{NW78}.

\mysubsection{Contribution and Techniques}
\label{sec:results}

Our main result is a $\frac{1}{6}$-approximation algorithm
for {\ggap} instances
where the total size of each group is at most $\frac{m}{2}$. 
We note that when group sizes can be arbitrary in $(0, m]$, {\ggap} cannot be approximated within any bounded ratio, even if 
item profits are identical across the bins,
and $m=2$, unless {\sf P}= {\sf NP}. Indeed, in this case, deciding whether a single group  of items of total size $2$ and total profit $1$ can be packed in the bins yields an instance of
{\sc Partition}, which is NP-complete~\cite{GJ79}. 
Furthermore, even if group sizes are restricted to be no greater than $\delta m$, for some $\delta> \frac{2}{3}$, then {\ggap} still cannot be approximated within a constant factor,
as it generalizes {\sf GMKP}$(\delta)$, for which the paper~\cite{CZ16} shows hardness of approximation.
Similarly, as we consider in this paper a generalization of {\sf GMKP}$\left(\frac{1}{2}\right)$, it follows from~\cite{CZ16} that our problem cannot be approximated within ratio
better than $\frac{1}{2}$.

In solving {\ggap} we combine
the framework of Adany et al.~\cite{AF+16} with the rounding technique
of Shmoys and Tardos~\cite{ST93}. The framework of~\cite{AF+16} uses
submodular maximization to select a collection of groups for the solution.
It then finds a feasible assignment for the selected groups.

At the heart of our algorithm lies 
an interesting derivation of a submodular
function from the classic LP formulation of
{\sf GAP}, which facilitates the construction of
a high profit solution utilizing
at most half the total bin capacity.
In particular, we give an algorithm for submodular maximization
subject to a knapsack constraint,
which finds a solution occupying
at most half the knapsack capacity,
while the other half is {\em reserved} for later use.\footnote{We assume throughout the discussion that
the size of each element is at most half the knapsack capacity.}
We show that this algorithm achieves an approximation ratio of $\frac{1}{3}$ relative to an optimal solution that may use the whole knapsack capacity.
We note that this ratio is tight. 
Indeed, it is easy to construct an instance
for which the best solution with half the knapsack capacity has only $\frac{1}{3}$ the profit of the optimal solution with full knapsack capacity. 
We also note that a naive application of the algorithm of Sviridenko~\cite{S04} with half the
knapsack capacity will only guarantee a
$\frac{1- e^{-1}}{3}\approx \frac{1}{4.7}$-approximation.

To obtain an integral solution, given a fractional assignment of the selected groups,
we apply the rounding technique of Shmoys and Tardos~\cite{ST93},
followed by a filling phase.
We show that if the total size of the items in the selected groups is
at most $\frac{m}{2}$, the rounding procedure yields a {\em feasible} assignment of the selected groups, whose profit is
at least half the value of the submodular function.
Our novel approach of submodular optimization 
subject to a knapsack {\em with reserved capacity} constraint
may find applications in solving other group assignment problems.

\section{Approximation Algorithm}
\label{sec:approx_alg}

In this section we present an approximation algorithm for {\ggap}. We first introduce several definitions and tools that will be used as
building blocks of our algorithm.

\subsection{Basic Definitions and Tools}
\label{sec:tools}
\subsubsection{The Submodular Relaxation}
\label{sub_intro}

For simplicity, we assume that all the numbers are rational.
For a subset of elements $I' \subseteq I$, let $s(I')= \sum_{i\in I'} s_i$ be the total size of the elements in $I'$.
We assume throughout the discussion that every $G_\ell \in {\cG}$
satisfies $s(G_\ell) \leq \frac{m}{2}$.
We define below a function ${\f}: 2^I \rightarrow {\Reals}^+\cup \{0\}$.
Let $x_{i,j} \in \{0,1\}$ be an indicator for the assignment of item $i$ to bin $j$, for $1 \leq i \leq n$, $1 \leq j \leq m$.
The following is a linear program associated with a subset $S \subseteq I$, in which $x_{i,j} \geq 0$ for all $i,j$.

\begin{alignat}{3}
LP(S)\text{:}&  \text{ maximize } &&\sum_{i\in I, j \in M} x_{i,j} \cdot p_{i,j} \nonumber\\
& \text{ subject to: } && \sum_{j \in M} x_{i,j}  \leq 1 && \forall{i\in I}  \label{select} \\
& && \sum_{i \in I} x_{i,j}  \cdot s_i \leq 1 && \forall{j \in M} \label{capacity} \\
& &&x_{i,j}=0 && \forall i \in I \setminus S, j \in M \nonumber \\
& && x_{i,j}\geq 0 && \forall i \in I, j \in M  \nonumber
\end{alignat}

Note that, by the above constraints, all solutions for the LP have the same dimension, regardless of the size of $S$.
We define ${\f}(S)$ as the optimal value of $\LP(S)$.

We denote the profit of a solution $x$ for the linear program by $p \cdot x = \sum_{i \in I, j \in M} x_{i,j} \cdot p_{i,j}$.
In Section~\ref{sec: submod} we prove the next result.

\begin{theorem}
	\label{sub}
	The function $\f$ is submodular.
\end{theorem}

We note that $\f$ is also monotone and non-negative. We use $\f$ to define the {\em group function} $\g: 2^{\cG} \rightarrow {\Reals}^+\cup \{0\}$. For any $G^* \subseteq {\cG}$ let
$I(G^*) = \bigcup_{G_\ell \in G^*} G_\ell$ and
$\g(G^*) = \f(I(G^*) )$. As $\f$ is submodular, monotone and non-negative, it is easy to see
that $\g$ is submodular, monotone and non-negative as well.
We optimize $\g$ subject to a knapsack (budget) constraint, using the next general result.

\begin{theorem} {\bf [Submodular optimization with reserved capacity]}
	\label{thm:halfsvi}
Let ${\ccG}= \{ 1, \ldots, n \}$ be a ground set,
and $m \geq 0$ a knapsack capacity.
Each $i \in {\ccG}$ is associated with non-negative size $s_i \leq \frac{m}{2}$. Let $f: 2^{\ccG} \rightarrow {\Reals}^+\cup \{0\}$ be a  non-negative monotone submodular function,
and $\OPT = \max \set{ f(S)| S\subseteq {\ccG}, \sum_{i\in S} s_i \leq m}$.
Then Algorithm \ref{alg:subds} (in Section \ref{sec:submod_optimize})
 finds in polynomial time\footnote{The explicit representation of a submodular function might be exponential in the size of its ground set. Thus, it is standard practice to assume that the function is accessed via a value oracle.
 Then the number of operations and oracle calls is polynomial in the size of ${\ccG}$ and the maximum length of the representation of $f(S)$.}
 a subset $S\subseteq {\ccG}$ satisfying
$f(S)\geq \frac{\OPT}{3}$ and $\sum_{i\in S} s_i \leq \frac{m}{2}$.
\end{theorem}
The proof of Theorem \ref{thm:halfsvi} is given in Section~\ref{sec:submod_optimize}.
\subsubsection{Solution Types}

Our algorithm uses a few types of intermediate solutions for {\ggap}, as defined below.
Given $G^* \subseteq {\cG}$, we say that a solution $x$ for $LP(I(G^*))$ is a {\em fractional solution}.
Let $U=(U_1,\ldots, U_m)$ be an assignment of elements to bins, where $U_j$ is the set of elements assigned to bin
$j$. Then $I(U) = \bigcup_{j=1}^{m} U_j$ is the subset of elements packed in the bins.
We say that $U$ is {\em feasible} if for each bin $1 \leq j \leq m$ we have $s(U_j) \leq 1$.
We say that $U$ is {\em almost feasible} if for each bin $1 \leq j \leq m$ there is an element $u_j^*$ such
that $s(U_j \setminus \set{u_j^*}) \leq 1$. We also define the profit of an assignment as
$p(U) = \sum_{1 \leq j \leq m} \sum_{i \in U_j} p_{i,j}$.

Our algorithm first obtains a fractional solution, which is then converted to an almost feasible solution. Finally, the algorithm converts this solution to a feasible one.
We now state the results used
in these conversion steps.

\begin{theorem}
	\label{thm:shmoys_trados}
	Given $G^* \subseteq {\cG}$, such that $s(I(G^*))\leq m$,  and a fractional solution $x$ for $LP(I(G^*))$, it is possible to construct in
	polynomial time an almost feasible assignment $U$ such that $p(U) \geq p \cdot x$, and
	$I(U) = I(G^*)$.
\end{theorem}

The theorem easily follows by applying a rounding technique of~\cite{ST93} to a fractional solution
in which every element in $I(G^*)$ is fully assigned (fractionally, in multiple bins). We note that such a solution always exists, since
$s(I(G^*)) \leq m$.
We give the proof in the Appendix.
To convert an almost feasible solution to a feasible one we use the following result (we give the proof in Section \ref{filling_sec}).

\begin{theorem}
	\label{thm:filling}
Let $U=(U_1, \ldots, U_m)$ be an almost feasible assignment
such that $s(I(U)) \leq \frac{m}{2}$, then $U$ can be converted in
polynomial time to a feasible assignment $U'$, with $I(U') = I(U)$
and $p(U') \geq \frac{1}{2} p(U)$.
\end{theorem}

\subsection{The Algorithm}

Our approximation algorithm for {\ggap} follows easily from
the tools presented in Section \ref{sec:tools}. Initially, we solve the problem of maximizing a submodular function
subject to a knapsack with reserved capacity constraint for the set function $\g$.
Then we solve the linear program and convert the
solution to a feasible assignment.
We give the pseudocode in Algorithm~\ref{alg:gg}.

\begin{algorithm}
	\caption{{\ggap} Algorithm}
	\label{alg:gg}
	\begin{algorithmic}[1]
		\State
		\label{step:submod_knapsack_constraint}
		Solve the
		submodular optimization problem:
		$\max_{ \{S\subseteq {\cG}, \sum_{G_\ell \in S} s(G_\ell) \leq m/2 \}} \g(S)$
		using  Algorithm \ref{alg:subds}. Let $S^*$ be the solution found.
		\State Find a (fractional) solution $x$ for $LP(I(S^*))$ that realizes $\g(S^*)$.
     \State Use Theorem~\ref{thm:shmoys_trados} to convert $x$ to an almost feasible assignment $U$ with $I(U) = I(S^*)$.
		\label{step:round}
		\State  Use Theorem~\ref{thm:filling} to convert $U$ into a feasible solution; return this solution.
		\label{step:feasible}
	\end{algorithmic}
\end{algorithm}

\begin{theorem}
	\label{thm:ggap_appx_ratio}
 Algorithm \ref{alg:gg} is a polynomial time  $\frac{1}{6}$-approximation algorithm for {\ggap}
 when the total size of a single group is bounded by $\frac{m}{2}$; that is, $\forall G_{\ell} \in {\cG}: \sum_{i \in G_{\ell} }s_i \leq \frac{m}{2} $.
\end{theorem}

\begin{proof}

It is easy to see that the algorithm runs in polynomial time.
By Theorem \ref{thm:halfsvi}, we have that $\g(S^*)\geq \OPT /3$, where $\OPT$ is the
value of the optimal solution for the original instance.

By Theorems \ref{thm:shmoys_trados} and \ref{thm:filling}, we are guaranteed to find in Steps~\ref{step:round}$-$\ref{step:feasible} a feasible assignment $U$ of all elements in $I(S^*)$,
such that $p(U) \geq \frac{1}{2} \g(S^*) \geq \frac{1}{6} \OPT$.
\end{proof}

\section{Submodularity}
\label{sec: submod}

In this section we show that the function $\f$ is submodular. Our proof builds on
the useful relation of our problem to maximum weight bipartite matching.
Let $G=(A \cup B,E)$ be a bipartite (edge) weighted graph, where $|B| \geq |A|$. Assume that the graph is complete (by adding zero weight edges if needed). For $e \in E$, let $W(e)$ be the weight of edge $e$, and for $F\subseteq E$,
let $W(F)= \sum_{e\in F} W(e)$ be the total weight of edges in $F$.
For $S\subseteq A$, define $h(S)$ to be the value of the maximum weight matching in $G[S \cup B]$, the graph induced by $S\cup B$. We call $h$ the {\em partial maximum
weight matching function}. The next result was shown by Bar-Noy and Rabanca~\cite{BR16}.

\begin{theorem}
\label{thm:submod_bp_matching}
If the edge weights are non-negative then the function $h$ is (monotone) submodular.
\end{theorem}
We give a simpler proof in the Appendix.
We are now ready to prove our main result.

\noindent
{\bf Proof of Theorem~\ref{sub}:}
We first note that since all numbers are rational, for some $N \in {\Int}^+$, we can write $s_i=\frac{\hsi}{N}$, where ${\hsi} \in {\Int}^+$ for all $i \in I$.\footnote{Note that $N$, which may be
	arbitrarily large, is used just for the proof. Our algorithm does not rely on obtaining a solution (or an explicit formulation) for $M(S)$.}

Now, set the capacity of each bin $1 \leq j \leq m$ to be $b_j=N$, and let $0  \leq y_{i,j} \leq {\hsi}$ indicate the size of item $i$ assigned to bin $j$.
For a subset of items $S \subseteq I$, we now write the following linear program.

\begin{alignat}{3}
M(S)\text{:}&  \text{ maximize } &&\sum_{i\in I} \frac{1}{\hsi} \sum_{j \in M} y_{i,j} \cdot p_{i,j} \nonumber\\
& \text{ subject to: } && \sum_{j \in M} y_{i,j}  \leq {\hsi} && \forall{i\in I}  \label{itemsize_bound} \\
& && \sum_{i \in I} y_{i,j}  \leq N && \forall{j \in M} \label{bincapacity_bound} \\
& && y_{i,j}=0 && \forall i \in I \setminus S, j \in M \nonumber \\
& && y_{i,j} \geq 0 && \forall i \in I, j \in M  \nonumber
\end{alignat}
Indeed, Constraint (\ref{itemsize_bound}) ensures that the total size assigned for item $i$ over the bins is upper bounded by ${\hsi}$, and Constraint (\ref{bincapacity_bound})
guarantees that the capacity constraint is satisfied for all the bins $j \in M$.
Given a subset of elements $S \subseteq I$,
let ${\ft}(S)$ be the value of an optimal solution for $M(S)$.

Now, observe that any feasible solution for $LP(S)$ induces a feasible solution for $M(S)$ of the same value, by setting
$y_{i,j}= x_{i,j} \cdot {\hsi}$ for all $i \in I$ and $j \in M$. Similarly, a feasible solution for $M(S)$ induces a feasible solution for $LP(S)$ of the same value.
Hence, $\f(S) = {\ft}(S)$ for all $S \subseteq I$.

By the above discussion, to prove the theorem it suffices to show that ${\ft}$
 is submodular. Given our {\ggap} instance, we construct the following bipartite graph $G$. For each item $i \in I$, we define ${\hsi}$ vertices,
 $V_i = \{ v_{i,1}, \ldots , v_{i, {\hsi}} \}$. For each bin $j \in M$, we define $N$ vertices $U_j = \{ u_{j,1}, \ldots , u_{j,N} \}$. For any $1 \leq i \leq n$ and $1 \leq j \leq m$, there are
 edges $(v_{i,s}, u_{j,r})$ of weight $p_{i,j}/{\hsi}$, for all $1 \leq s \leq {\hsi}$, $1 \leq r \leq N$. Let $V_I = \cup_{i \in I} V_i$, $U_M = \cup_{j \in M} U_j$, and
 let $E$ be the set of edges.
 Consider the bipartite graph $G=(V_I \cup U_M, E)$. W.l.o.g we may assume that $|U_M| \geq |V_I|$; otherwise, we can add new bins $j=m+1, m+2, \ldots$ with corresponding sets of vertices $U_j = \{ u_{j,1}, \ldots , u_{j,N} \}$ and zero weight edges $(v_{i,s}, u_{j,r})$ for all $1 \leq i \leq n$, $1 \leq s \leq {\hsi}$, $1 \leq r \leq N$.

 We note that, given a subset of items $S \subseteq I$, $M(S)$ is the linear programming relaxation of the problem of finding a maximum weight matching in the
 subgraph $G[V_S \cup U_M]$,
 where $V_S \subseteq V_I$ is the subset of vertices in $G$ that corresponds to $S$.
 Using standard techniques (see, e.g.,~\cite{LRS11}), it can be shown that $M(S)$ has an optimal integral solution.
 Hence, ${\ft}(S)=h(V_S)$, where $h: 2^{V_I} \rightarrow {\Reals}^+$ is a partial maximum weight matching function in $G$.
 By Theorem~\ref{thm:submod_bp_matching}, $h$ is (monotone) submodular.
 Hence, ${\ft}$ is also (monotone) submodular. 	\hfill \qed

\newcommand{\km}{5}
\newcommand{\kv}{6}
\section{Submodular Optimization with Reserved Capacity}
\label{sec:submod_optimize}

In this section we prove Theorem~\ref{thm:halfsvi}. We start with some definitions and notation.
Assume we are given a ground set
${\ccG}  = \{ 1, \ldots, n \}$ and capacity $m > 0$, where  each element $i \in {\ccG}$ is associated with a non-negative size $s_i \leq \frac{m}{2}$.
For $S \subseteq {\ccG}$, let $s(S) = \sum_{i\in S} s_i$.
Also, for $S,T \subseteq {\ccG}$ let $f_S(T) = f(S\cup T) - f(S)$. We use throughout this section basic properties of
monotone submodular functions (see, e.g.,~\cite{BF17}).
\negA
\begin{algorithm}
	\caption{\sc SubmodularOpt} 
	\label{alg:subds}

	\hspace*{\algorithmicindent} \noindent \textbf{Input:} A monotone submodular function $f:2^{\ccG}\rightarrow \mathbb{R}_+$,
	sizes $s_i \geq 0$ for all $i \in {\ccG}$, and capacity $m > 0$.
	
	\hspace*{\algorithmicindent} \noindent \textbf{Output:} A subset of elements $R \subseteq {\ccG}$ such that $s(R) \leq \frac{m}{2}$. 
	

	\begin{algorithmic}[1]	
		
		\Procedure{Greedy}{$g$, $m'$}
		\State Set $S= \emptyset$, $E={\ccG}$.
		\While {$E\setminus S \neq \emptyset$}
		\State
		Find
		$i = \argmax_{i \in E \setminus S}  \frac{g_S(\{i\})}{s_i}$
		\label{greedy:argmax}
		\If {$s(S)  + s_{i} \leq m'$}
		set $S = S \cup \{i\}$.
		\EndIf
		\State Set $E = E \setminus \{i\}$.

		\EndWhile
		\State Return $S$
		\EndProcedure
		
		\State
		Set $R= \emptyset$
		\For {every set  $S_e \subseteq {\ccG}$, 
			$|S_e|\leq \kv$}
		\For {every set $B\subseteq S_e$, $s(B)\leq m/2$}
		\State $T = \Call{Greedy}{f_{S_e}, m/2 - s(B)}$
		\If {$f(B \cup T)\geq f(R)}$
		Set $R = B \cup T$.
		\EndIf
		\EndFor
		\EndFor
		\State Return $R$
	\end{algorithmic}
\end{algorithm}

In the following we give an outline of an algorithm for maximizing a monotone submodular function $f$ subject to 
a knapsack {\em with reserved capacity} constraint.
Specifically, assuming that the knapsack capacity is $m$ for some $m >0$, the algorithm solves the problem $\max_{ \{ S \subseteq {\ccG}: s(S) \leq \frac{m}{2}\} } f(S)$.
The algorithm, {\sc SubmodularOpt}, initially guesses the set of at most
six items of highest profits in some optimal solution (for the problem with knapsack capacity $m$),
and a subset of these profitable items, whose total size is at most $m/2$.
Then the algorithm calls a procedure which applies the Greedy approach as in~\cite{S04} to find the remaining items in the solution. We give a pseudocode of {\sc SubmodularOpt}
in Algorithm~\ref{alg:subds}.  It is important to note that while the algorithm produces
a solution of size at most $m/2$, the analysis compares this solution against an optimal solution of size at most $m$. 

The next lemma, which plays a key role in our analysis, 
follows from the technique  presented in \cite{S04}. 

\begin{lemma}
\label{thm:svi}
Given the knapsack capacity $m > 0$, let $0 < m' \leq m^* \leq m$.
Let $S^* \subseteq {\ccG}$ be a non-empty subset of elements, such that $s(S^*) \leq m^*$. Also, let
$g: 2^{\ccG} \rightarrow \Reals$
be a monotone submodular function satisfying $g(\emptyset)=0$,  and let $S=\Call{Greedy}{g,m'}$.
Then, there is an element $i^*\in S^*$ such that
$g(S)+g(\{i^*\})\geq (1-e^{m'/m^*}) g(S^*)$.
\end{lemma}

In the Appendix we prove a more general result (see Lemma~\ref{thm:alpha_greedy}).\footnote{The statement of Lemma~\ref{thm:svi} is obtained by taking in Lemma~\ref{thm:alpha_greedy} 
$T=\emptyset$.}
In \cite{S04} the above lemma was applied in conjunction with a guessing phase, used to ensure the three most
profitable elements in an optimal solution are selected by the 
algorithm, thus limiting the value of $g(\{i^*\})$
in the special case where $m'=m^*$.

Attempting to apply a similar approach in solving the problem with 
{\em reserved} capacity, several difficulties arise.
The first one is that 
the most profitable elements in {\em any} optimal solution may already
exceed the reduced capacity, and therefore cannot be added to the solution.
Another difficulty is that even if these elements do fit in the smaller knapsack, 
one can easily come up with a scenario 
in which it is better not to include them in the solution.

To overcome these difficulties we use the following main observation.
Given $P_k$, the set of $k=\kv$ most profitable elements 
in an optimal solution\footnote{The value $k=6$ is derived from Lemma \ref{thm:sub_technical}. 
	It may be possible to obtain the same approximation ratio using smaller values of $k$, 
	leading to a more efficient algorithm.}, and  a partition of this set into two subsets $B_1, B_2$, each of size at most $m/2$, adding elements
to either $B_1$ or $B_2$ using the greedy procedure leads
to a solution 
of value at least one third of
the value of an optimal solution. This observation comes into play
in Case 2.2 in the proof of Theorem~\ref{thm:halfsvi}.
The next technical lemma  
is used to prove this observation (we give the proof below).

\begin{lemma}
	\label{thm:sub_technical}
	For $k=\kv$, $p_A, p_B, S_A, S_B \geq 0$ define
	$$h(p_A,p_B,S_A, S_B) = p_A+ (1-p_A-p_B) \left(1-e^{ - \frac{\frac{1}{2} - S_A}{1- S_A - S_B} }\right) - 
	\frac{p_A+p_B} {k}.$$
	Then for $p_1, p_2, S_1, S_2$ such that $0\leq p_1, p_2 \leq \frac{1}{3}$ and $0\leq S_1, S_2 \leq \frac{1}{2}$
	it holds that 
	$$\max\left( h(p_1, p_2, S_1, S_2), h(p_2, p_1, S_2, S_1)\right) \geq \frac{1}{3}.$$
\end{lemma}

Another main tool used in the proof of Theorem~\ref{thm:halfsvi} is a simple partitioning procedure. 
It shows that $P_k$ can either be partitioned into two sets as 
required in the above observation, or we reach a simple corner case (Case 2.1 in the proof) in which at least one third of the optimal value can be easily attained. For the latter case, we use
the following result, due to~\cite{LB10}.

\begin{lemma}
	\label{thm:lin}
	Let $g: 2^{\ccG}\rightarrow \Reals$ be a
	non-negative and monotone submodular function.
	Let $\OPT = \max \set{ g(S)| S\subseteq {\ccG}, \sum_{i\in S} s_i \leq m}$, and $S=\Call{Greedy}{g, m}$. Then either $g(S) \geq (1-e^{-1/2}) \OPT$,
	or there is an element $i\in {\ccG}$ such that $g(\{i\})\geq (1-e^{-1/2})OPT$ and $s_i \leq m$.
\end{lemma}

\noindent
{\bf Proof of Theorem \ref{thm:halfsvi} [Submodular optimization with reserved capacity]:}
It is easy to see that the running time of the algorithm is polynomial.
Let $S \subseteq {\ccG}$, $s(S)\leq m$, $f(S)= \OPT$, and $k=6$.

\noindent
{\bf \underline {Case 1:}}
	We first handle the case where $|S| < k$.
	Start with $A_1 = \emptyset$, iterate over the elements of
	$S$ and add them to $A_1$, as long as $s(A_1) \leq m/2$.
	If $S\neq A_1$, let $j\in S\setminus A_1$, and set $A_2= \{j\}$
	and $A_3 =S \setminus (A_1 \cup A_2)$.  Clearly, $s(A_2 ) \leq m/2$,
	and since $s(A_1 \cup A_2 ) > m/2$ and $s(S)\leq m$, we have that $s(A_3)\leq m/2$.
	If $S = A_1$  set  $A_2 = A_3 =\emptyset$.
	
	By the submodularity of $f$, we have $f(S)\leq f(A_1) + f(A_2) + f(A_3)$. Hence,
	for some $r\in \{1,2,3\}$, $f(A_r) \geq f(S)/3=\OPT/3$.  We also have
	that $|A_r| \leq 5$; therefore, at some iteration of the algorithm $S_e =B=A_r$, and following this iteration $f(R)\geq \OPT/3$.

\noindent
{\bf \underline{Case 2:}}
	Assume now that $|S| \geq k$. Let $S= \{i_1, i_2, \ldots, i_\ell\}$ such that the
	elements are ordered by their marginal profits:
	$i_j = \argmax_{j\leq r \leq \ell} f_{\{i_1, \ldots i_{j-1}\}}(\{i_r\})$. Set
	$P_k = \{i_1, i_2, \ldots , i_k\}$.
	
	Consider the following process.  Start with $B_1 = \emptyset$ and $B_2 = \emptyset$. Iterate over the elements $i\in P_k$ in decreasing order by size. For each element $i$, let $t = \argmin_{j=1,2} s(B_j)$.
	If $s(B_t)+s_{i} \leq m/2$ then $B_t = B_t \cup \{i\}$;  otherwise, Stop.
	We now distinguish between two sub-cases for the termination of the process.
	

\noindent
{\bf \underline{Case 2.1:}}
		Suppose that the process terminates due to an element $i$ which cannot be added to any of the sets.
		Let $B_1$ and $B_2$ be the sets  in this iteration. Also, set $B_3 = \{i\}$, $U= B_1 \cup B_2 \cup B_3$,
		and $L = S\setminus U$. W.l.o.g assume that $s(B_1) \geq s(B_2)$.
		As the process terminated, we have that $s(B_3)+s(B_2) = s_i + s(B_2) >m/2$. The sets
		$B_1, B_2, B_3$ and $L$  form a partition of  $S$, and $s(S)\leq m$.
		We conclude that $s(B_1) + s(L) \leq m/2$.
		Hence, $s(B_2) +s(L) \leq m/2$, and $s(B_3) + s(L) \leq m/2$
		as well (it is easy to see that $s(B_3)\leq s(B_1)$).
		
		By the submodularity of $f$, $f(U)\leq f(B_1)+f(B_2) +f(B_3)$; thus,
		there is $j\in\{1,2,3\}$ such that $f(B_j)\geq f(U)/3$.
		As none of the sets $B_1,B_2,B_3$ is empty, we have that $|B_j| \leq |P_k| -2=4$.
		
		Let $T = \Call{Greedy}{f_{U},m/2 - s(B_j)}$. By Lemma \ref{thm:lin}, either
		\[
		\begin{array}{c}f_{U} (T) \geq (1-e^{-1/2}) f_{U}(L) \geq f_{U}(L)/3,
		\end{array}
		\]
		or there is $i\in L$ such that
		\[
		\begin{array}{c}f_{U} (\{i\}) \geq (1-e^{-1/2}) f_{U}(L) \geq f_{U}(L)/3.
		\end{array}
		\]
		
		In the former case, we can consider the iteration in which
		$S_e =U, B=B_j$. In this iteration, we have  $$f(B \cup T) \geq f(B_j) + f_{U}(T) \geq
		\frac{1}{3} ( f(U)+ f_U(L)) = \frac{1}{3}\OPT.$$
		In the latter case, we can consider the iteration where
		$S_e  = B=B_j \cup \{i\}$, and in which
		$$f(B \cup T) \geq f(B) \geq f(B_j) + f_{B_j} ( \{i\}) \geq f(B_j) + f_{U} ( \{i\}) \geq \frac{1}{3} \OPT.$$

\noindent
{\bf \underline{Case 2.2:}}	
		The process terminated with $B_1, B_2$  satisfying $B_1 \cup B_2= P_k$, and $s(B_1),s(B_2)\leq m/2$. 
		
		Let $p_1 = f(B_1) /\OPT$, $p_2 = (f(P_k)- f(B_1)) /\OPT$, $S_1 = s(B_1)$, $S_2 = s(B_2)$ and $L = S\setminus P_k$.
		If $p_1\geq \frac{1}{3}$ (or $p_2 \geq \frac{1}{3}$ ) we have that in the iteration
		where $S_e= P_k$  and $B=B_1$ (or $B = B_2$) the algorithm finds a solution of value at least
		$\OPT/3$, and the theorem holds. Thus, we may
		assume that $p_1, p_2 \leq \frac{1}{3}$.
		
		For $j=1,2$, let $T_i = \Call{Greedy}{f_{P_k} , m/2-S_j}$. Using
		Lemma \ref{thm:svi} with $S^* = L$, 
		we have that there is $i_j\in L$  for which
		$$f_{P_k}(T_j)+f_{P_k}(\{i_j\}) \geq  (1-e^{-\frac{\frac{1}{2}- S_j}{ 1- S_1 - S_2}}) f_{P_k} (L).$$
		By the selection of elements in $P_k$, we have 
		$f_{P_k}(\{i_j\})\leq \frac{1}{k} f(P_k)$. Thus,
		$$f_{B_j}(T_j)\geq f_{P_k}(T_j)\geq  (1-e^{-\frac{\frac{1}{2}- S_j}{ 1- S_1 - S_2}}) f_{P_k} (L) - \frac{1}{k} f(P_k).$$
		Hence, in the iteration where $S_e = P_k, B=B_j$, we obtain a solution satisfying
		\[
		\begin{array}{ll}
		f(S_e\cup T) & = f(B_j \cup T_j) \geq
		f(B_j) + (1-e^{-\frac{ \frac{1}{2} - S_j}{1-S_1 -S_2 }}) f_{P_k} (L) - \frac{1}{k} f(P_k) \\ \\
	   & \geq \displaystyle{ \OPT \left( p_j +(1-p_1+p_2) (1-\exp( -\frac{ \frac{1}{2} - S_j}{1-S_1 -S_2 } ) ) - \frac{p_1+p_2}{k}\right).}
		\end{array}
		\]
		By Lemma \ref{thm:sub_technical}, 
		in one of these iterations we obtain a solution of value at least $\OPT/3$, implying the statement of the theorem. \hfill \qed

\bigskip
\noindent
{\bf Proof of Lemma \ref{thm:sub_technical}:}
Let $p_1, p_2, S_1, S_2$ be values that satisfy the conditions in the lemma. 
Denote $p = p_1 + p_2$, $d= p_1 - p_2$ and $r = \exp\left(-\frac{\frac{1}{2} - S_1} {1 -S_1 -S_2}\right)$.
It is easy to see that $\exp\left(-\frac{\frac{1}{2} - S_2} {1 -S_1 -S_2}\right) = \frac{e^{-1}}{r}$.
Define $V_1 = h(p_1, p_2, S_1, S_2)$ and $V_2 = h(p_2, p_1, S_2, S_1)$.  By the definition of $h$ and above 
definitions we get 
$$V_1 = \frac{p+d}{2} +(1-p)(1-r) - \frac{p}{k}$$
and 
$$V_2 = \frac{p-d}{2} +(1-p)\left(1-\frac{e^{-1}}{r}\right) - \frac{p}{k}$$
Let $g_1(x) = \frac{p+d}{2} +(1-p)(1-x) - \frac{p}{k}$ 
and $g_2(x) = \frac{p-d}{2} +(1-p)\left(1-\frac{e^{-1}}{x}\right) - \frac{p}{k}$. 
Clearly, $V_1 = g_1(r)$ and $V_2 = g_2(r)$. It is also easy to 
see that $g_1$ is decreasing and $g_2$ is increasing (for $x>0$).
\begin{claim}
	It holds that $g_1(x^*) = g_2(x^*)$ where  $x^*= \frac{d + \sqrt{d^2 + 4e^{-1} (1-p)^2}}{2(1-p)}$.
	
\end{claim} 
\begin{proof}
	By rearranging terms we have  $g_1(x) = g_2(x)$ 
	if and only if $d = (1-p) ( x -\frac{e^{-1}}{x})$, which holds
	for $x>0$ if and only if
	$0 = (1-p) x^2 -dx -e^{-1}(1-p)$. The latter is a quadratic
	equation and $x^*>0$ is a root.  
\end{proof}

If $r\geq x^*$, since $g_2$ is increasing, we have $V_2 = g_2(r) \geq g_2(x^*) = g_1(x^*)$, and if $r\leq x^*$, as $g_1$ is decreasing, we have $V_1 = g_1(r) \geq g_1(x^*)$. Therefore $\max (V_1, V_2) \geq g_1(x^*)$. 

Consider two cases:

\noindent
{\bf Case 1:} $p\leq \frac{1}{3}$. We have $|d|\leq p$  and
$$g_1(x^*) = 1-\frac{p}{2} -\frac{\sqrt{d^2 + 4e^{-1}(1-p)^2}}{2} - \frac{p}{k} \geq 1-  \frac{p}{2} - \frac{p}{k} - \frac{\sqrt{p^2 + 4e^{-1}(1-p)^2}}{2}.$$

Let $a(y) = 1- \frac{y}{2} - \frac{y}{k} - \frac{\sqrt{y^2 + 4e^{-1}(1-y)^2}}{2}$. Clearly, $g_1(x^*)\geq a(p)$. 
We note that $a''(y)= -\frac{2e^{-1}} {( y^2+ 4e^{-1} (1-y)^2 )^{\frac{3}{2}}} $.
As $a''(y) < 0$, $a(y)$ does not have 
a local minima in $[0, \frac{1}{3}]$. Thus as $p\in [0, \frac{1}{3}]$
we have $$g_1(x^*)\geq a(p) \geq \min\left( a(0), a\left(\frac{1}{3}\right)\right) = \min \left( 1-e^{-0.5}, \frac{5}{6}- \frac{1}{3k} -\frac{\sqrt{1+16e^{-1}}}{6}\right) \geq \frac{1}{3}$$

\noindent
{\bf Case 2:} $p\geq  \frac{1}{3}$, as $\frac{p+d}{2}, \frac{p-d}{2} \leq \frac{1}{3}$ we have that $|d| \leq \frac{2}{3}- p$. Therefore,
$$g_1(x^*) = 1-\frac{p}{2} -\frac{\sqrt{d^2 + 4e^{-1}(1-p)^2}}{2} - \frac{p}{k} \geq 1-  \frac{p}{2} - \frac{p}{k} - \frac{\sqrt{(\frac{2}{3}-p)^2 + 4e^{-1}(1-p)^2}}{2}$$
Let  $b(z) = 1- \frac{z}{2} - \frac{z}{k} - \frac{\sqrt{(\frac{2}{3} - z)^2 + 4e^{-1}(1-z)^2}}{2}$. Clearly, $g_1(x^*) \geq b(p)$.  We note 
that $$b''(z) = -\frac{2e^{-1}}{9 \left( 4e^{-1}(1-z)^2 +(\frac{3}{2}-z )^2 \right) ^{\frac{3}{2}}}<0.$$
Thus,  $b(z)$ does not have a local minima in $[\frac{1}{3}, \frac{2}{3}]$. Hence, we have that $g_1(x^*) \geq b(p) \geq \min(b(\frac{1}{3}),b( \frac{2}{3})) \geq \frac{1}{3}$.

In both cases we get $g_1(x^*)\geq\frac{1}{3}$, and as $\max(V_1, V_2)\geq g_1(x^*)$, the lemma immediately follows.
\hfill \qed

\section{Filling Phase}
\label{filling_sec}

In this section we prove Theorem~\ref{thm:filling}. 
Define the {size of bin} $j$ in assignment $U$ as  $s^U_j = \sum_{i\in U_j} s_i$.
We first divide the bins and items into types.
We say that a bin $j$ is {\em full}  if $s^U_j > 1$, {\em semi-full} if $ \frac{1}{2} \leq s^U_j \leq 1$, and
{\em semi-vacant} if $s^U_j < \frac{1}{2}$.
An item $i \in I(U)$ is {\em big} if $s_i > \frac{1}{2}$; otherwise,
$i$ is {\em small}. Clearly, there are no big items in semi-vacant bins.

Informally, we use in the proof several types of resolution steps. Each step takes as input
a full bin and possibly one or two semi-vacant bins, and reassigns some of the items into the bins while evicting others.
These resolution steps ensure that the new assignment has at least half the profit of the original assignment, the assignment to any bin remains feasible, and only small items
are evicted.

We apply the resolution steps repeatedly,  but once a bin participated
in a resolution step it may not participate in another one.
We then prove that as long as there are full bins, one of the steps can be
applied. Hence, by applying the resolution steps, we have a new
assignment in which all bins are feasible, and the total profit is at
least half the profit of the original assignment.
To handle the evicted items, we note that
as $s(I(U))\leq m/2$  and all the evicted items are small, it
is possible to assign the evicted items to bins without
violating the capacity constraints.

\bigskip

\noindent
{\bf Proof of Theorem~\ref{thm:filling}:}
For any bin $1 \leq j \leq m$ and $A \subseteq I$, let $p_j(A)  = \sum_{i\in A} p_{i,j}$ be
the total profit gained from packing $A$ into $j$.
The first step is to resolve the violation of the capacity
constraint in each full bin. We do that using four types of
resolution steps. Each step takes a full bin and possibly one or two semi-vacant bins,
modifies their contents and adds some small elements to a set $V$ of {\em evicted}
elements, that will be handled later. Throughout the discussion, we consider for a full bin $j$ a partition of the elements in $U_j$ into
two feasible subsets, given by $\set{A_j,B_j}$. We use the following resolution steps.

\begin{enumerate}
	\item
	\label{nobig}
	Consider a full bin $j$ such that $U_j$ has no big elements.
	If $p_j(A_j) > p_j(B_j)$ then set $U'_j = A_j$ and evict $B_j$ ($V : = V \cup B_j$);
	otherwise, set $U'_j = B_j$ and evict $A_j$, i.e., $V : =V \cup A_j$.
	In both cases $U'_j$ is feasible and $p_j(U'_j) \geq \frac{1}{2} \cdot p_j(U_j)$.
	\item
	\label{one_big}
	Now, suppose we have a full bin $j$ such that $U_j$ has a single big element, and
	a semi-vacant bin $\ell$.
	Let $\set{A^*, B^*} = \set{A_j,B_j}$, such that the big element is in $A^*$.
	
	If $p_j(A^*) +p_\ell(U_\ell) > p_j(B^*)$, set $U'_j = A^*$, $U'_\ell= U_\ell$
	and evict the elements in $B^*$ ($V := V \cup B^*$). We note that in this
	case $p_j(U'_j) + p_\ell(U'_\ell) = p_j(A^*) +p_\ell(U_\ell)$.
	
	Otherwise, set $U'_j = B^*$ and $U'_\ell= A^*$, and evict all the elements
	in $U_\ell$, i.e., $V := V \cup U_\ell$. In this case we have
	$p_j(U'_j) + p_\ell(U'_\ell) \geq p_j(B^*)$.
	
	Therefore, in both cases have  $p_j(U'_j) + p_\ell(U'_\ell) \geq \frac{1}{2} \cdot \left(p_j(U_j) + p_\ell(U_\ell)\right).$
	
	\item \label{two_big}
	Consider a full bin $j$ such that $U_j$ has two big elements, and
	a semi-vacant bin $\ell$, such that one of the big elements has space
	in bin $\ell$; that is, there is a big element $i^* \in U_j$ such that
	$s_{i^*} + s^U_{\ell} \leq 1$.
	
	Let  $\set{A^*, B^*} = \{ A_j,B_j \}$ such that $i^*\in A^*$.  We note that
	there cannot be any other big element in $A^*$ other that $i^*$.
	
	If $p_j(B^*) + p_\ell(U_\ell) > p_j(A^*)$ set  $U'_j = B^*$ and
	$U'_\ell = U_\ell \cup \set{i^*}$ (note that $s^{U'}_{\ell} \leq 1$). Also,
	evict all elements in $A^* \setminus \set{i^*}$.
	In this case we have $p_j(U'_j) + p_\ell(U'_\ell) \geq p_j(B^*)+ p_\ell(U_\ell)$.
	
	Otherwise, we set $U'_j = A^*$ and $U'_\ell = B^*$, and evict
	$U_\ell$ ($V := V \cup U_\ell$). In this case we have
	$p_j(U'_j) + p_\ell(U'_\ell) \geq p_j(A^*)$.

	Thus, in both cases
	$p_j(U'_j) + p_\ell(U'_\ell) \geq \frac{1}{2} \cdot \left(p_j(U_j) + p_\ell(U_\ell)\right).$
	
	\item \label{last_resort}
	Finally, consider a full bin $j$, such that $U_j$ has two big elements, and
	two semi-vacant bins $\ell_1$ and $\ell_2$. 
	Recall $A_j,B_j$ is a partition of the elements in $U_j$ into two feasible subsets.
	
	If $p_j(A_j)  + p_{\ell_1}  (U_{\ell_1})> p_j(B_j) + p_{\ell_2} ( U_{\ell_2})$,
	set $U'_j = A_j$ , $U'_{\ell_1} = U_{\ell_1}$  and $U'_{\ell_2} = B_j$, and
	evict $U_{\ell_2}$.  Thus, we have
	$$p_j(U'_j) + p_{\ell_1}(U'_{\ell_1}) + p_{\ell_2}(U'_{\ell_2}) \geq
	p_j(A_j) + p_{\ell_1} ( U_{\ell_1}).$$
	
	Otherwise, set $U'_j = B_j$, $U'_{\ell_1} = A_j$ and $U'_{\ell_2} = U_{\ell_2}$ and evict
	$U_{\ell_1}$. Here $$p_j(U'_j) + p_{\ell_1}(U'_{\ell_1}) + p_{\ell_2}(U'_{\ell_2}) \geq
	p_j(B_j) + p_{\ell_2} ( U_{\ell_2}).$$
	
	And finally, we always have
	\[p_j(U'_j) + p_{\ell_1}(U'_{\ell_1}) + p_{\ell_2}(U'_{\ell_2}) \geq
	\frac{1}{2} \cdot \left(  p_j(U_j) + p_{\ell_1}(U_{\ell_1}) + p_{\ell_2}(U_{\ell_2})  \right) .
	\]
	
\end{enumerate}

Each time we execute a step we mark the bins used in this step
as {\em resolved}, and we do not consider them in the next steps.
We first use Steps \ref{nobig},\ref{one_big}, and \ref{two_big}, until none of them can
be applied.

Consider the average size of the unresolved bins.
When we start, there are $m$ bins of average size no greater than half.
Each of Steps  \ref{nobig},\ref{one_big}, \ref{two_big}  reduces
the size of the unresolved bins by at least one (as a full bin is removed)
and reduces the number of bins by at most two. Therefore, the
average size of the unresolved bins remains no more than half.
Also, marking all the semi-full bins as resolved  preserves the property.

Let $a$ be the number of unresolved full bins  and $c$ be the number of
unresolved semi-vacant bins.  Due to the average size of bins,
we have $a \leq \frac{a+c}{2}$, therefore $a \leq c$.
Hence, if we have a full bin, there must be a semi-vacant bin as well.
As we used Steps \ref{nobig},\ref{one_big}, \ref{two_big}  to exhaustion,
every full bin must have two big elements (no bin in $U$ contains more than two big elements),
and none of these big elements can fit into
one of the semi-vacant bins.

Denote the minimal size of a semi-vacant unresolved bin by $r$. Then
each of the full bins has two big elements of size greater than $1-r$.
Hence, we have
$c\cdot r + 2  a \cdot ( 1-r) < \frac{a+c}{2}$,
which leads to
$ 2 a(1-r) - \frac{a}{2}  < c\left(\frac{1}{2} - r \right)$,
and

\[
\begin{array}{ll}
c &  > a \cdot \displaystyle{\frac{2-2r - \frac{1}{2}}{\frac{1}{2} -r }=
	a \cdot \frac{3-4r  }{1 -2r }} \\
\\
& = \displaystyle{
	a \cdot \left( \frac{2-4r}{1-2r} + \frac{1}{1-2r}\right) =
	a \cdot \left( 2+ \frac{1}{1-2r}\right)  > 2a},
\end{array}
\]
implying that we can now run Step \ref{last_resort}, until there are no more
unresolved full bins.

We use the resolution steps to eliminate all the full bins. Every time
we run such a step over a set of bins we lose at most half the profit of
the bins participating in the step.
As the total size of items in the assignment is bounded by $m/2$, we
are guaranteed that it is possible to assign all the evicted elements to some
bins.
Thus, we are able to resolve the capacity overflow while losing at most
half of the profit. \hfill \qed

\section{Discussion and Future Work}

In this paper we presented a $\frac{1}{6}$-approximation algorithm for {\ggap}, using  a mild restriction
 on group sizes. A key component in our result is 
an algorithm for submodular maximization subject to a knapsack constraint,
which finds a solution occupying
at most half the knapsack capacity, while the other half is {\em reserved} for later use.
Our results leave several avenues for future work.

As mentioned above, {\ggap} with no assumption on group sizes cannot be
approximated within any constant factor. 
Yet, the maximum group size that still allows to obtain a 
constant ratio can be anywhere in $[\frac{m}{2}, \frac{2}{3}m]$.
A natural question is whether our results can be applied to instances with larger group sizes. We note that the ratio stated in Theorem~\ref{thm:ggap_appx_ratio} 
may not hold already for instances in which group sizes can be at most $\frac{m}{2}(1 + \eps)$, for some $\eps >0$. Indeed, for such instances, 
it may be the case that no set of groups of total size at most $m/2$ is `good' relative to the optimum. 
The existence of an algorithm that yields a constant ratio for such instances remains open.

While our result for submodular optimization with reserved capacity (Theorem \ref{thm:halfsvi}) gives an optimal approximation ratio for the studied subclass of instances, we believe the result 
can be extended to other subclasses. In particular, we conjecture that for instances where each item is of size at most $\delta>0$, the approximation ratio approaches
$1-e^{-\frac{1}{2}}$ as $\delta \rightarrow 0$.
Such result would yield improved the approximation ratio for instances of {\ggap} in which the total size of each 
group is bounded by $\delta m$. We defer this line of work to the full version of the paper.

Lastly, we introduced in the paper the novel approach of submodular optimization subject to a knapsack with reserved capacity constraint. We applied the approach along with  
a framework similar to the one developed in \cite{AF+16}. It is natural 
to ask whether the approach can be used to improve the approximation ratio 
obtained in \cite{AF+16} for {\em all-or-nothing} {\sf GAP}.

\bibliography{ggap}
\newpage
\appendix
\section{Appendix}
\label{sec:appendix}
\newcommand{\cM} { {\cal M}}

\subsection{Applications of {\sf Group GAP}}
\label{sec:appendix_applics}
As a primary motivation for our study, consider wireless networks~(see, e.g., \cite{KPCK14,WZSL17}) a carrier supplies a set of data items to a set of users upon demand.
Each user has an isolated cache memory of a given size; each item
consists of a set of blocks. Time is slotted. The carrier smooths out the network load by caching at time $t=0$ data block within the user caches, and later lets the users
share these blocks at time $t \geq 1$ via a D2D (Device-to-Device) network. There is a utility associated with storing a data block of an item in each user cache (depending, e.g., on the location of this user in the network).
Thus, viewing each item as a group of blocks, finding a feasible allocation of data blocks in the users caches so as to maximize the total utility yields an instance of {\sf Group GAP}.

{\ggap} has a central application also in 5G networks.
Caching at small-cell base stations (SBSs) to meet the overwhelming growth in mobile data demand is expected to
be utilized in the next generation (5G) cellular networks~\cite{PT16,BBD14}.
This technique leverages the expected deployment of short-range, low-power, and low-cost SBSs in cellular networks.
It is expected that this will enable immersive Virtual Reality (VR) and Augmented Reality (AR) mobile applications that are currently limited to off-line settings~\cite{Chak17}.
A VR/AR $360^{\circ}$-navigable 3D video presents
the user with an interactive immersive visual representation of the
remote scene of interest that the user can navigate from any possible
3D viewpoint. To enable such capability in a mobile device with limited computation power and storage, the 3D video is decomposed into different
views and {\em all} these views must be cached in the SBSs.
There is a utility associated with storing a specific view in each SBS (depending, e.g., on the projected popularity of this view, its size relative to the cache size, etc.).
Thus, viewing each 3D video as a group of views, finding a feasible allocation of these views to SBSs to maximize the total utility yields an instance of
{\sf Group GAP}.

\subsection{Some Proofs}
\label{sec:appendix_proofs}
\noindent
{\bf Proof of Theorem \ref{thm:shmoys_trados}:}
Let $S=I(G^*)$. 
We first note that since $s(S)=s(I(G^*)) \leq m$, any fractional solution $x$ can be modified to a solution $x'$ of at least the same profit, such that for all
$i\in S$ it holds that $\sum_{j\in M} x'_{i,j} = 1$.
Thus, we may assume that
for any $i\in S$ we have $\sum_{j\in M} x_{i,j} =1$.
W.l.o.g., assume  the items are ordered by their size; that is,
$s_1 \geq s_2 \geq \ldots \geq s_n$.
We use Algorithm~\ref{alg:st_construct} to generate
a bipartite graph $G = (S, R, E)$ with weight $w_e \geq 0$ on any edge $e \in E$,
and a fractional matching $x'$ in $G$.

We note the following observations about Algorithm~\ref{alg:st_construct}.

\begin{itemize}
	\item
	The algorithm returns a fractional matching $x'$ which is complete for $S$.
	For each item $i\in S$ and bin $j\in M$ it holds that
	$x_{i,j} = \sum_{e=(i,(j,r) )\in E} x'_e$. Therefore,
	the weight of $x'$ is equal to the profit of the solution $x$ for
	the  $\LP$, i.e., $w(x') = p\cdot x$.
	
	\item  For any bin $j\in M$ let $k_j$ be the maximum number such
	that $(j,k_j) \in R$. For any $1\leq r \leq k_j$ let
	$s_{j,r} ^{\max} = \max\{ s_i | (i,(j,r))\in E\}$ and
   $s_{j,r} ^{\min} = \min\{ s_i | (i,(j,r))\in E\}$.
   By the construction, it is easy to see that
   $s_{j,1}^{\max}\geq s_{j,1}^{\min}\geq s_{j,2}^{\max} \geq s_{j,2}^{\min}\geq \ldots
   s_{j,k_j}^{\max}\geq s_{j,k_k}^{min}$
	\item For any bin $j\in M$ and $1\leq r <k_j$ we have $\sum_{i|(i,(j,r))\in E} x'_{(i,(j,r))} = 1$.
\end{itemize}

Using standard results from matching theory
 (see, e.g.,~\cite{LP09}), there is a matching $\cM$ for $G$
which is complete for $S$  and $\sum_{e\in M} w_e \geq w(x')$. Such
a matching $\cM$ can also be found in polynomial time.

We define an assignment based on $\cM$; that is,
for every bin $j\in M$ and $1\leq r \leq k_j$ set
$U_{j,r}= \{i\in S| (i,(j,r))\in \cM\}$. Now, define
for every $j\in M$: $U_j =\bigcup_{r=1}^{k_j} U_{j,r}$ and $U = (U_1, U_2,\ldots, U_m)$.
As $\cM$ is complete for $S$, each element $i\in S$ is assigned to a single bin,
and each of the sets $U_{j,r}$ contains at most a single element.
Now, for every $j \in M$ we have
\begin{equation}
\label{eq:almost_feasible}
\begin{aligned}
\displaystyle{ \sum_{i \in U_j} s_i }  & =  \displaystyle{\sum_{r=1}^{k_j} \sum_{i \in U_{j,r}} s_i   
 \leq 		
 \sum_{i\in U_{j,1}} s_i + \sum_{r=2}^{k_j} s_{j,r}^{\max} 
   \leq 
   \sum_{i\in U_{j,1}} s_i  + \sum_{r=1}^{k_j-1} s_{j,r}^{\min}}
  \\  \\
  &=  \displaystyle{  \sum_{i\in U_{j,1}} s_i  + \sum_{r=1}^{k_j-1}\sum_{\{i| (i, (j,r))\in E\}} x'_{(i,(j,r))}  s_{j,r}^{\min} }
  \\ \\
  & \leq
 \displaystyle{\sum_{i\in U_{j,1}} s_i  + \sum_{r=1}^{k_j-1}\sum_{\{i| (i, (j,r))\in E\}} x'_{(i,(j,r))} s_i 
  \leq  
   \sum_{i\in U_{j,1}} s_i  + \sum_{i\in S} x_{i,j} s_i}
  \\ \\
  & \leq 
  \displaystyle{1+ \sum_{i\in U_{j,1}} s_i}
\end{aligned}
\end{equation}
The
first inequality follows from the definition of $s_{j,r}^{\max}$ and $|U_{j,r}| \leq 1$.
The second inequality holds since $s_{j,r}^{\max} \leq s_{j,r-1}^{\min}$.
The second equality holds since $\sum_{i|(i,(j,r))\in E} x'_{(i,(j,r))} = 1$ for $1\leq r<k_j$.
In the third inequality we use the definition of $s_{j,r}^{\min}$, and
the fourth  inequality follows from $x_{i,j} = \sum_{e=(i,(j,r) )\in E} x'_e$.
The last inequality holds since $x$ is a (feasible) solution for $\LP(S)$.

\begin{algorithm}
	\caption{Bipartite Construction}
	\label{alg:st_construct}
	
	\begin{algorithmic}[1]	
		\State Set $R=E=\emptyset$
		\For {$j\in M$}
		\State Set $r=1$ and $c_r=0$, add $(j,1)$ to $R$
		\For {$i$ from $1$ to $n$ such that $x_{i,j}>0$}
		\State Add an edge $e=(i,(j,r))$ to $E$, $w_e =p_{i,j}$
					and set $x'_e = \min(x_{i,j}, 1-c_r)$
		\State Update $c_r = c_r + x'_e$.
		\If{$x_{i,j} \neq x'_e$}
			\State Add $(j,r+1)$ to $R$ and $f= (i,(j,r+1))$ to $E$.
			 \State Set  $w_f = p_{i,j}$, $x'_f  = x_{i,j} -  x'_e$, $r=r+1$ and  $c_r = x'_f$.
		\EndIf
		\EndFor
		
		\EndFor
		\State Return $G=(S, R, E)$, the weights $w$ and the
		fractional matching $x'$.
	\end{algorithmic}
	
\end{algorithm}

By inequality \eqref{eq:almost_feasible} and since $U_{j,1}$ contains at most a single element, we have that $U$ is an almost
feasible assignment.  We note that the profit of the assignment $U$ is
exactly the weight of the matching $\cM$, which is at least the profit
of $x$.

This suggests the following conversion procedure. Given a solution $x$ for $\LP(S)$, use Algorithm \ref{alg:st_construct} to obtain a bipartite graph $G= (S, R, E)$ with weights $w$. Find a
maximal weight matching $\cM$ in $G$ which is complete for $S$. Define an assignment
 $U$ using $\cM$ and return this assignment.
 By the above discussion, this procedure is polynomial and returns an assignment
 as required by the theorem.
\qed

\noindent
{\bf Proof of Theorem~\ref{thm:submod_bp_matching}:}
Consider two subsets $S\subset T \subset A$. Let $v\in A\setminus T$. To prove submodularity we show that $h(S\cup\{v\})-h(S) \ge h(T\cup\{v\})-h(T)$.

Let $S' = S\cup\{v\}$ and $T' = T\cup\{v\}$. For a subset $X\subseteq A$ let $M_X$ be a matching whose weight achieves $h(X)$. (Note that there may be more than one such matching, we fix one arbitrarily.) For $u \in X$, define $M_X(u)$ to be the vertex in $B$ that is matched to $u$.
Also, for $Y\subseteq X$ let $M_X(Y)$ be the set of matching {\em edges} that touch the vertices in $Y$. That is, $M_X(Y) =\cup_{y\in Y} \{(y,M_X(y))\}$.

To prove the theorem we need the following lemma.
\begin{lemma}
\label{lem:augPath}
There exists a subset $X=\{x_1=v,\ldots,x_k\}\subseteq S'$, for some $k\ge 1$,
with the following two properties.
\begin{enumerate}
\item
For $i\in [1..k-1]$, $M_{S}(x_{i+1}) = M_{T'}(x_{i})$
\item
$M_{T'}(x_k)$ is not matched in $M_S$
\end{enumerate}
\end{lemma}
Before we prove the lemma, we show how it implies the theorem.
Let $E_1 = M_S(X\setminus\{x_1\})$, and $E_2 = M_{T'}(X)$.

We claim that $h(S')-h(S) \ge W(E_2)-W(E_1)$. Consider the set of edges $E_2
\cup M_S(S\setminus X)$. Notice that it is a matching in the subgraph induced by
$S'\cup B$. Clearly, $h(S') \ge W(E_2 \cup M_S(S\setminus X)) =
W(E_2)+W(M_S(S\setminus X))$. Also, by definition,
$h(S)=W(E_1\cup M_S(S\setminus X)) = W(E_1)+W(M_S(S\setminus X))$.
Thus, $h(S')-h(S) \ge W(E_2)-W(E_1)$.

We claim that $h(T')-h(T) \le W(E_2)-W(E_1)$. By definition,
$h(T')=W(E_2\cup M_{T'}(T\setminus X)) = W(E_2)+W(M_{T'}(T\setminus X))$.
Consider the set of
edges $E_1 \cup M_{T'}(T\setminus X)$. Notice that it is a matching in the
subgraph induced by $T\cup B$. Clearly,
$h(T) \ge W(E_1 \cup M_{T'}(T\setminus X)) = W(E_1)+W(M_{T'}(T\setminus X))$.
Thus, $h(T')-h(T) \le W(E_2)-W(E_1)$. The proof follows.
\hfill \qed

\begin{dl_proof}{\ref{lem:augPath}}
We build $X$ iteratively. We start with $X=\{v\}$. If $M_{T'}(v)$ is not
matched in $M_S$, we are done since $X$ satisfies both properties. (The first
property is satisfied vacuously.) Otherwise, define $x_2$ to be the vertex that
satisfies  $M_S(x_2) = M_{T'}(x_1=v)$. Note that $x_2\ne x_1$. Also,
$X=\{x_1,x_2\}$ satisfies the first property by definition. If $M_{T'}(x_2)$ is
not matched in $M_S$, we are done since $X$ satisfies both properties.
Otherwise, define $x_3$ to be the vertex that satisfies   $M_S(x_3) =
M_{T'}(x_2)$. Note that $x_3 \notin \{x_1,x_2\}$. We continue in the same manner
if we are not done. Namely, at stage $i$ we have the set $\{x_1,\ldots,x_i\}$
that satisfies the first property. If $M_{T'}(x_i)$ is not matched in $M_S$, we
are done. Otherwise, define $x_{i+1}$ to be the vertex that satisfies
$M_S(x_{i+1}) = M_{T'}(x_i)$.  Again, $x_{i+1} \notin \{x_1,\ldots,x_i\}$. Since
in each stage we add a new vertex from $S$,  and $S$ is finite, we are guaranteed
to be done at some stage.
\end{dl_proof}

\subsection{Missing Details in the Proof of Theorem \ref{thm:halfsvi}}
\label{app:sec4_details}
\begin{lemma}
	\label{thm:alpha_greedy}
	Let $S^* \subseteq {\ccG}$ be a non-empty subset of items, such that $s(S^*) \leq m^*$. Also, let
	$g: 2^{\ccG} \rightarrow \Reals$
	be a submodular, non-negative and monotone function such that $g(\emptyset)=0$,  and $S=\Call{Greedy}{g,m'}$.
	Then, for any subset $T \subseteq {\ccG}$,
	there is an item $i^*\in S^*$ such that
	$g(S)+g_T(\{i^*\})\geq (1-e^{m'/m^*}) g_T(S^*)$.
\end{lemma}

\begin{proof}
	If $S^* \subseteq S$ then the claim trivially holds. Thus, we may assume
	that $S^* \nsubseteq S$.
	Let $S = \{i_1, i_2, \ldots, i_\ell\}$, where the elements are in the order they
	were added to $S$ during the execution of {\sc Greedy}.
	Let $S_r =\{ i_1, i_2, \ldots, i_r \}$ denote the first $r$ items in $S$, $0 \leq r \leq \ell$, where $S_0=\emptyset$. We denote by 
	$\theta_r = \frac{g_{S_{r-1}} (\{i_r\})}{s_{i_r}}$ the marginal gain density from adding $\{i_r\}$ to $S_{r-1}$.
	Given a subset $T \subseteq {\ccG}$, we define 
	$\eta_r = \max_{i\in S^* \setminus S_{r-1}} \frac{g_T(S_{r-1} \cup \{i\})- g_T(S_{r-1})}{s_i}$
	for $1\leq r \leq \ell +1$
	(since $S^* \nsubseteq S$ the maximization is never over an empty set).
	By the definition of $\theta_r$ we have 
	$g(S_r) = \sum_{j=1}^{r} \theta_j s_{i_j}$ 
	for
	$1\leq r \leq \ell$.
	If there is $1\leq r \leq \ell$ for which $\eta_r >\theta_r$,
	set $t = \min\{r| 1\leq r \leq \ell, \eta_r >\theta_r\}-1$, otherwise
	set $t = \ell$.
	
	Let $i^* \in S^* \setminus S_t$ such that $\eta_{t+1} = \frac{g_T(S_t \cup \{i^*\})- g_T(S_t)}{s_{i^*}}$ (one exists by the definition of $\eta_r$). If $t<\ell$, by the
	submodularity of $g$ we have that
	$\frac{g_{S_t}(\{i^*\})}{s_i} \geq \eta_{t+1} > \theta_{t+1}$; therefore
	in the iteration of {\sc Greedy} in which $i_{t+1}$ was added to $S$, it holds that $i^* \notin E$ (otherwise
	it would have been selected instead of $i_{t+1}$), and we can conclude
	that $s(S_t)+ s_{i^*} > m'$. If $t=\ell$, we know that by the end of
	the execution of {\sc Greedy} $i^* \notin E$, therefore $s(S_t)+ s_{i^*} > m'$ as well.
	
	Define $m_r =s(S_r)$ for $0\leq r \leq t$ and  $m_{t+1} = s(S_t) + s_{i^*}$. Also,
	define $p_j = \theta_r$ for $1\leq r \leq t$ and $j = m_{r-1} +1 , m_{r-1}+2, \ldots, m_r$,
	and $p_j = \eta_{t+1}$ for $j = m_t +1, m_t +2, \ldots, m_{t+1}$.
	Clearly, by the submodularity of $g$, 
	it holds that $g_T(S_r)\leq g(S_r) = \sum_{j=1}^{m_r} p_j$, for $0 \leq r \leq t$.
	Also, it is easy to see that for any $0\leq r \leq  t$  we have $\eta_{r+1} \leq p_{m_r+1}$.
	
	The following sequence of inequalities follows from the above observation
	and the submodularity of $g$. For any $0\leq r \leq t$:
	\[
	\begin{array}{ll}
	g_T(S^*)   & \leq   g_T(S_r) + \sum_{i \in S^* \setminus S_r} \left(g_T(S_r \cup \{i\})-g_T(S_r) \right)  \\ \\
	& \leq   g_ T(S_r) + m^* \eta_{r+1}  \leq \sum_{j=1}^{m_r} p_j + m^* p_{m_r +1}
	\end{array}
	\]
	Therefore,
	\[
	\begin{array}{l}
	g_T(S^*) \leq \min_{1\leq r \leq t }\left\{\sum_{j=1}^{m_r} p_j + m^* p_{m_r +1}\right\}=
	\min_{1\leq s \leq m_{t+1} } \left\{ \sum_{j=1}^{s-1} p_j + m^* p_s\right\}
	\end{array}
	\]
	
	We now use inequality (4) in~\cite{S04}:
	
	\[
	\begin{array}{ll}
	\frac{g(S_t) + g_T(\{i^*\})}{g_T(S^*)}
	& \geq
	\frac{ \sum_{s=1}^{m_{t+1}} p_j } { \min_{1\leq s \leq m_{t+1} } \left\{ \sum_{j=1}^{s-1} p_j + m^* p_s\right\}} \\  \\
	& \geq
	(1-e^{m_{t+1}/m^*}) \geq (1-e^{m'/m^*}),
	\end{array}
	\]
	
	and the lemma follows immediately. 
\end{proof}

\end{document}